\documentclass[a4paper,11pt]{article}

\usepackage[T1]{fontenc}
\usepackage[utf8]{inputenc}
\usepackage{fullpage}
\usepackage{times}
\usepackage{microtype}
\usepackage{amsmath,amssymb,amsthm,mathtools}
\usepackage{booktabs}
\usepackage{float}
\usepackage{enumitem}
\usepackage[small]{caption}
\usepackage{authblk}
\usepackage{natbib}
\usepackage[hidelinks]{hyperref}
\usepackage[capitalize,nameinlink]{cleveref}
\usepackage{xcolor}

\newtheorem{theorem}{Theorem}[section]

\newtheorem{proposition}[theorem]{Proposition}
\newtheorem{lemma}[theorem]{Lemma}

\theoremstyle{definition}
\newtheorem{definition}[theorem]{Definition}
\newtheorem{example}[theorem]{Example}

\allowdisplaybreaks
\usepackage{tikz}
\usepackage{pifont}
\usepackage{multirow,multicol}

\usetikzlibrary{shapes,decorations,arrows,calc,arrows.meta,fit,positioning}
\tikzset{
    -Latex,auto,node distance =1 cm and 1 cm,semithick,
    state/.style ={ellipse, draw, minimum width = 0.7 cm},
    point/.style = {circle, draw, inner sep=0.04cm,fill,node contents={}},
    bidirected/.style={Latex-Latex,dashed},
    el/.style = {inner sep=2pt, align=left, sloped}
}

\title{\bf Full Justified Representation in Temporal Voting: \\ Efficient Computation and Verification}

\author[]{Nicholas Teh}
\affil[]{University of Oxford, UK}

\date{\vspace{-10mm}}

\begin{document}

\maketitle
\begin{abstract}
In temporal voting, a fixed set of voters makes a collective decision in each of several rounds, and proportionality requires that groups with shared interests be represented across these decisions. Full justified representation (FJR) is among the strongest proportionality axioms known to be satisfiable in this setting, since it allows different members of a group to be satisfied by different selected candidates. However, whether an FJR outcome can be computed in polynomial time has remained open. We resolve this question by introducing temporal FJR+, a strengthening of both FJR and extended justified representation+ (EJR+) that can be computed and verified in polynomial time. The axiom measures a group's agreement in each round by the largest number of its members who approve a common candidate, and requires some member to attain the integer part of the group's proportional share of its best attainable average satisfaction, with a shortfall of at most one unit when this share is an integer. We show that verification reduces to counting approvals among the voters below each satisfaction threshold. For computation, we observe that local search for Proportional Approval Voting (PAV) can stop at an outcome violating FJR. Somewhat counterintuitively, the remedy is to subtract a small multiple of the minimum voter satisfaction from the PAV score: every FJR+ violation then admits a single-round change whose improvement guarantees polynomial running time. We further show that temporal FJR+ is strictly stronger than FJR and EJR+ combined, and coincides with lower quota on party-list preferences. Finally, temporal FJR+ and multiwinner FJR+ are incomparable in general, but coincide when approvals are static and each candidate-round pair is treated as a separate candidate.\end{abstract}

\section{Introduction}
\label{sec:introduction}

Suppose that an AI assistant is asked to plan a week of activities for a group of people. Each participant indicates which of the available activities she would enjoy on each day. A natural approach is to select, every day, the activity approved by the most participants. However, the same majority may then prevail every day, leaving some participants with nothing they enjoy all week. Considering the decisions jointly creates an opportunity for \emph{proportional representation}: a group with shared interests should receive a share of the programme commensurate with its size and with the extent of its agreement. Shared interests, however, need not mean identical preferences. Within a group of culturally inclined participants, some may prefer concerts and others exhibitions; a good programme gives each of them several activities they like, even if no single activity pleases them all.

\emph{Temporal voting} captures this kind of sequential decision making. A fixed electorate selects one candidate in each of several rounds, voters' approval preferences may change from round to round, and a voter's satisfaction is the number of selected candidates she approves. The ballots may be elicited directly or supplied by learned preference models.\footnote{\citet{chandak2026proportional} make the latter
connection concrete by aggregating votes from models trained on Moral
Machine responses from different countries.} We study the \emph{offline} setting, in which all approval ballots are available before any decision is made.

How should proportionality treat a group whose members agree only partially? The familiar representation axioms from multiwinner voting start from groups whose members approve common candidates \citep{aziz2017justified}. \emph{Full justified representation} (FJR), introduced by \citet{peters2021proportional} for participatory budgeting, is more permissive: it allows different members of a group to be satisfied by different parts of the collective choice. Its temporal version asks what every member of a group could obtain if the group controlled a number of rounds proportional to its size \citep{phillips2025strengthening}. The group specifies a set of rounds, and its benchmark is the satisfaction it can guarantee to all of its members within \emph{every} subset of these rounds of proportional size. The benchmark therefore reflects both the group's size and its opportunities for agreement, while preventing the group from reserving only its most favourable rounds. FJR then requires that at least one member of the group reach this benchmark in the selected outcome.

An outcome satisfying temporal FJR always exists. However, the known construction repeatedly searches over groups and possible decisions \citep{phillips2025strengthening} and does not yield a polynomial-time algorithm. Verification faces a similar obstacle, since the definition quantifies over all groups, all sets of rounds, and all possible decisions within them. In contrast, temporal extended justified representation+ (EJR+) is both efficiently computable and efficiently verifiable \citep{phillips2025strengthening}. This contrast raises our central question: \emph{Can we retain the flexibility of FJR while making both computation and verification efficient?}

We answer this question affirmatively by introducing \emph{temporal FJR+}, an axiom that implies both FJR and EJR+ and admits polynomial-time computation and verification.

\subsection{Our Contributions}
We work in the standard offline model of temporal voting, with approval ballots that may vary across rounds, and establish the following results.
\begin{itemize}
\item \textbf{A stronger representation guarantee.}
FJR+ is strictly stronger than FJR and EJR+ combined, even when every
voter approves exactly one candidate per round
(\cref{thm:implications,ex:strict}). On static party-list preferences,
it is equivalent to lower quota: each party wins at least the
integer part of its proportional share of the rounds (\cref{prop:party}).

\item \textbf{Polynomial-time verification.}
Although the axiom quantifies over all groups, it suffices to inspect,
for each satisfaction threshold, the group of all voters below that
threshold: enlarging a group preserves its cohesiveness and cannot decrease
its agreement counts. Verification therefore reduces to counting
approvals (\cref{thm:verification}), and a failed check identifies an
underrepresented group.

\item \textbf{Polynomial-time computation.}
We modify local search for Proportional Approval Voting (PAV) by
subtracting a small multiple of the minimum voter satisfaction from the
PAV score (\cref{thm:computation}). Like ordinary local search, the
rule changes one decision at a time, but the perturbation allows it to
leave outcomes at which ordinary local search can stop despite an FJR
violation. We show that every FJR+ violation admits an improvement in
the modified score that is large enough to ensure polynomial running
time.
\end{itemize}

For $n$ voters, $m$ candidates, and $\ell$ rounds, verification takes
$O((n+\ell)m\ell)$ time and computation takes
$O(n^2m\ell^5\log(\ell+1))$ time. Both results allow empty approval
ballots and require no divisibility assumptions.
Finally, we compare temporal FJR+ with multiwinner FJR+: neither implies
the other in general, but they coincide under static preferences when
each candidate-round pair is treated as a separate candidate
(\cref{sec:multiwinner}).

\subsection{Related Work}
\citet{lackner2020perpetual} introduced perpetual voting, and
\citet{lackner2023proportional} studied proportionality for online
decisions. \citet{bulteau2021justified} adapted justified representation
(JR) and proportional justified representation (PJR) to temporal voting,
and \citet{chandak2026proportional} established temporal representation
guarantees for several familiar voting rules, including local search
PAV. Our starting point is the pair of temporal FJR and EJR+ axioms of
\citet{phillips2025strengthening}. Our verification algorithm builds on
their observation that enlarging a group preserves cohesiveness, and our
computation refines single-round PAV comparisons to handle the equality
cases that the stronger guarantee requires. Efficient verification is a
substantive requirement in its own right: \citet{elkind2025verifying}
proved that verifying temporal JR, PJR, and extended justified
representation (EJR) is coNP-complete.

The offline model is a special case of public decision making with
additive utilities, in which rounds play the role of issues
\citep{conitzer2017fair,alouf2022better}. For binary issues,
\citet{skowron2022public} give proportionality guarantees to groups
whose members need not agree on every issue. Our axiom instead uses
the actual agreement counts in each round and allows more than two
candidates. \citet{masarik2024generalised} study proportionality under
general feasibility constraints; as noted by
\citet{phillips2025strengthening}, their FJR existence theorem also
implies the existence of temporal FJR outcomes. We provide
polynomial-time computation and verification for the stronger temporal
FJR+ axiom.

In multiwinner voting, \citet{brill2023robust} introduced EJR+ as an
efficiently verifiable strengthening of EJR, and recent work by
\citet{frank2026polynomial} and \citet{ai2026full} gives
polynomial-time algorithms for FJR. \citet{teh2026strengthening}
introduces multiwinner FJR+, which can be verified in polynomial time
through linear programming and computed through a rule based on voter
budgets. That axiom measures agreement through fractional voter and
candidate weights, whereas ours measures agreement within individual
rounds. Moreover, multiwinner algorithms do not directly produce valid
temporal outcomes: an unconstrained committee of candidate-round pairs
may contain several pairs from one round and none from another. In
participatory budgeting, \citet{neoh2026strengthening} obtain
polynomial-time computation of FJR up to one project for additive
utilities.

\citet{elkind2024temporalelections} study welfare and strategyproofness
together with an individual proportionality guarantee, under which each
voter receives at least the integer part of her maximum possible
satisfaction divided by the number of voters. Related individual
guarantees have been studied in online/temporal fair division, where goods
arrive over time and must be allocated irrevocably
\citep{choi2026tfdmm,choo2026approxproponline,cohen2026budgetconstraints,elkind2025temporal,neoh2026online,Goldberg2026,neoh2026closinggaps}.
Other temporal models concern changes between successive committees
\citep{zech2024multiwinnerchange} and the fair assignment of projects
to time slots \citep{elkind2022temporalslot}; see
\citet{elkind2024temporal} for a survey.

\section{Preliminaries}
\label{sec:preliminaries}

For a positive integer $k$, let $[k]=\{1,\ldots,k\}$. A
\emph{temporal election} is a tuple $E=(C,N,\ell,A)$, where $C$ is a
nonempty finite set of $m$ candidates, $N=[n]$ is a nonempty set of
voters, $\ell\geq1$ is the number of rounds, and
$A=(a_{i,r})_{i\in N,r\in[\ell]}$ collects the approval sets
$a_{i,r}\subseteq C$, which may be empty. Preferences are
\emph{static} if each voter approves the same set in every round. An
\emph{outcome} is a vector
$\mathbf{o}=(o_1,\ldots,o_\ell)\in C^\ell$ selecting one candidate per
round; the same candidate may be selected in several rounds. For
$R\subseteq[\ell]$, we write $\mathbf{o}_R$ for the restriction of
$\mathbf{o}$ to the rounds in $R$, and voter $i$'s satisfaction with it
is $\mathrm{sat}_i(\mathbf{o}_R)=|\{r\in R:o_r\in a_{i,r}\}|$. In
particular, satisfaction over the empty set of rounds is zero, and we
write $\mathrm{sat}_i(\mathbf{o})$ for satisfaction over all rounds.
Unless stated otherwise, all groups $S\subseteq N$ considered below are
nonempty.

Our measure of a group's agreement is defined round by round. For a
group $S\subseteq N$ and a round $r$, let
$h_r(S)=\max_{c\in C}|\{i\in S:c\in a_{i,r}\}|$ be the largest number
of members of $S$ who approve a common candidate in round $r$; we also
set $h_r(\varnothing)=0$. The same quantity appears in each round of
the EJR+ verification algorithm of \citet{phillips2025strengthening}.
Following that work, a group $S$ is \emph{$(\sigma,\tau)$-cohesive},
for $\sigma\in[n]$ and $\tau\in[\ell]$, if $h_r(S)\geq\sigma$ in at
least $\tau$ rounds. Note that the voters who agree may differ from
round to round.

\begin{samepage}
We use the following two representation axioms of
\citet{phillips2025strengthening}.

\begin{definition}[EJR+]
An outcome $\mathbf{o}$ satisfies \emph{extended justified
representation+} (EJR+) if, for every $(\sigma,\tau)$-cohesive group
$S\subseteq N$ and every round $r$ with
$\bigcap_{i\in S}a_{i,r}\neq\varnothing$, either (i) $\mathrm{sat}_i(\mathbf{o})\geq
\left\lfloor\frac{\sigma\tau}{n}\right\rfloor$ for some $i \in S$, or (ii) $o_r\in\bigcap_{i\in S}a_{i,r}$.
\end{definition}
\end{samepage}

In other words, as long as every member of a cohesive group remains
below the satisfaction level justified by the group's cohesiveness, the
outcome must select a candidate approved by all members in every round
in which such a candidate exists.

FJR instead measures what a group could guarantee to all of its members
by controlling its proportional share of the decisions. The group first
specifies a set of rounds $T$, of which it is entitled to
$\lfloor |T||S|/n\rfloor$. For each subset of $T$ of this size, the
group may choose one candidate per round so as to maximize the
satisfaction of its least satisfied member, and its benchmark is the
smallest value obtained over all such subsets. Formally, writing $C^R$
for the set of choices of one candidate for each round in $R$, let
\[
\mu_S(T)=
\min_{R\subseteq T, |R|=\lfloor |T||S|/n\rfloor}
\ \max_{\mathbf{x}_R\in C^R}
\ \min_{i\in S}\mathrm{sat}_i(\mathbf{x}_R).
\]

\begin{definition}[FJR]
An outcome $\mathbf{o}$ satisfies \emph{full justified representation}
(FJR) if, for every group $S\subseteq N$ and every
$T\subseteq[\ell]$, some voter $i\in S$ has
$\mathrm{sat}_i(\mathbf{o})\geq\mu_S(T)$.
\end{definition}

The minimum over $R$ is essential: a group cannot reserve only the
rounds most favourable to it. On the other hand, the group may choose
different candidates for different subsets $R$, and its members need
not approve the same selected candidates.

\section{Full Justified Representation\texorpdfstring{\,}{}+}
\label{sec:axiom}

FJR allows the members of a group to approve different selected
candidates, but it measures the group's entitlement by what all members
can attain simultaneously within every proportional-size subset of
rounds. Our axiom starts instead from the group's total attainable
satisfaction over all rounds. This total may be large even when it can
only be distributed unevenly among the members; as in FJR, we use it to
demand representation for at least one member.

Recall that in the counts $h_r(S)$, both the maximizing candidate and
the voters approving it may change between rounds. A
$(\sigma,\tau)$-cohesive group has at least $\sigma$ such approvals in
each of at least $\tau$ rounds, and hence at least $\sigma\tau$
approvals in total. Summing the actual counts captures more: it also
accounts for agreement beyond $\sigma$ in these rounds and for
agreement in the remaining rounds. Since choices in different rounds
are independent, the largest total satisfaction the group can attain is
$\sum_{r=1}^{\ell}h_r(S)
=\max_{\mathbf{x}\in C^\ell}\sum_{i\in S}\mathrm{sat}_i(\mathbf{x})$.
Dividing this quantity by $n$ gives it a proportional interpretation:
\begin{equation}
\frac{1}{n}\sum_{r=1}^{\ell}h_r(S)
=\frac{|S|}{n}\left(
\max_{\mathbf{x}\in C^\ell}
\frac{1}{|S|}\sum_{i\in S}\mathrm{sat}_i(\mathbf{x})\right).
\label{eq:proportional-average}
\end{equation}
That is, it equals the group's share of the electorate multiplied by
its best attainable average satisfaction.

For example, suppose that $n=10$ and a group of six voters has round
counts $(4,3,3,2)$. Its best total satisfaction is $12$, its best
average satisfaction is $2$, and its proportional share of this average
is $1.2$. Requiring some member to have satisfaction at least one thus
draws on the group's agreement in all four rounds. Cohesiveness alone does
not impose this requirement: the largest product $\sigma\tau$ for which
the group is $(\sigma,\tau)$-cohesive is $9<n$.
A natural requirement is therefore
\begin{equation}
\max_{i\in S}\mathrm{sat}_i(\mathbf{o})
\geq\left\lfloor\frac{1}{n}\sum_{r=1}^{\ell}h_r(S)\right\rfloor
\quad\text{for every nonempty }S\subseteq N.
\label{eq:stronger-total}
\end{equation}
We show in \cref{prop:integer-boundary} that this requirement can
always be satisfied, but whether a satisfying outcome can be computed
in polynomial time remains open. Our axiom, which we can compute
efficiently, requires some member to reach every integer level
strictly below the proportional average. It additionally retains the
levels justified by $(\sigma,\tau)$-cohesiveness. For the whole
electorate, there are no competing groups among which the rounds must
be divided, so the axiom also includes the integer part of its best
attainable average.

\begin{definition}[FJR+]
\label[definition]{def:fjrplus}
An outcome $\mathbf{o}$ satisfies \emph{full justified representation+}
(FJR+) if, for every group $S\subseteq N$ and every target satisfaction
level $\beta\in[\ell]$,
some voter $i\in S$ has $\mathrm{sat}_i(\mathbf{o})\geq\beta$
whenever at least one of the following holds:
\begin{enumerate}[label=(\roman*),leftmargin=*,itemsep=3pt]
\item $\sum_{r=1}^{\ell}h_r(S)>n\beta$;
\item $S$ is $(\sigma,\tau)$-cohesive for some
$\sigma\in[n]$ and $\tau\in[\ell]$ with $\sigma\tau\geq n\beta$;
\item $S=N$ and $\sum_{r=1}^{\ell}h_r(S)\geq n\beta$.
\end{enumerate}
\end{definition}

The three conditions differ only in the boundary case
$\sum_r h_r(S)=n\beta$. Indeed, $(\sigma,\tau)$-cohesiveness implies
$\sum_r h_r(S)\geq\sigma\tau$, so condition~(i) already applies
whenever $\sigma\tau>n\beta$. Conditions~(ii) and~(iii) therefore only
determine which integer levels are protected when the total is exactly
$n\beta$. In this case, condition~(ii) holds precisely when all
positive values $h_r(S)$ are equal: $\tau$ rounds each contribute
exactly $\sigma$, and every other round contributes zero. This
observation will make verification particularly simple.

Consequently, FJR+ imposes the same target as \eqref{eq:stronger-total}
whenever the proportional average is not an integer. When it is a
positive integer $\beta$, FJR+ still guarantees satisfaction at least
$\beta-1$, and conditions~(ii) and~(iii) may require $\beta$. The two
requirements thus differ by at most one unit, and only at exact
equality.

Two further features of the axiom are worth noting. First, it does not
require groups to be specified in advance, and it does not depend on
how the outcome was computed. Second, since the three conditions depend
only on the election, any change that weakly increases every voter's
satisfaction preserves FJR+. An outcome is \emph{Pareto efficient} if
no other outcome weakly increases every voter's satisfaction and
strictly increases the satisfaction of at least one voter; such a
change is a \emph{Pareto improvement}. Because there are finitely many
outcomes, every FJR+ outcome is weakly improved upon by some
Pareto-efficient outcome, which then satisfies FJR+ as well.

\subsection{Relation to FJR and EJR+}

We chose the integer levels in FJR+ so that they include every level
justified by temporal FJR. The key step converts FJR's universal
requirement over subsets of rounds into a lower bound on the group's
total attainable satisfaction. When this bound holds with equality,
the group has the cohesiveness required by condition~(ii) if $S\neq N$,
while condition~(iii) applies if $S=N$.

\begin{theorem}
\label{thm:implications}
Every outcome satisfying FJR+ satisfies both FJR and EJR+.
\end{theorem}

\begin{proof}
For EJR+, let $S$ be $(\sigma,\tau)$-cohesive and set
$\beta=\lfloor\sigma\tau/n\rfloor$. If $\beta=0$, the requirement is
met trivially. Otherwise, condition~(ii) gives a voter in $S$ with
satisfaction at least $\beta$. Hence FJR+ always yields the
satisfaction alternative of EJR+, even when the group has no round of
complete agreement.

For FJR, fix $S$ and $T$, and write $q=\lfloor |T||S|/n\rfloor$. If
$\mu_S(T)=0$, there is nothing to show, so suppose that
$\beta=\mu_S(T)\geq1$; necessarily $1\leq\beta\leq q$. By the
definition of $\mu_S(T)$, every $q$-element subset $R\subseteq T$
admits a choice of candidates giving each member of $S$ at least
$\beta$ approvals. Therefore
\[
\sum_{r\in R}h_r(S)\geq\beta|S|
\quad\text{for every such }R.
\]
Each round in $T$ belongs to a fraction $q/|T|$ of these subsets, so
averaging these inequalities and using $q\leq|T||S|/n$ gives
\begin{equation}
\sum_{r=1}^{\ell}h_r(S)
\geq\sum_{r\in T}h_r(S)
\geq\frac{|T|\beta|S|}{q}
\geq n\beta.
\label{eq:fjr-average}
\end{equation}
If the first quantity exceeds $n\beta$, condition~(i) applies, and if
$S=N$, condition~(iii) applies. It remains to consider a group
$S\neq N$ with $\sum_r h_r(S)=n\beta$.

In this case, every inequality in \eqref{eq:fjr-average} holds with
equality; in particular, $|T||S|=nq$. The sums
$\sum_{r\in R}h_r(S)$ over all $q$-element subsets $R\subseteq T$
have average $\beta|S|$, and each of them is at least $\beta|S|$, so
every such sum equals $\beta|S|$. Since $S\neq N$, we have
$1\leq q<|T|$. Given two distinct rounds in $T$, choose $q-1$ other
rounds of $T$. Adjoining either of the two rounds yields a
$q$-element subset, and since the two subsets have the same sum, the
two rounds have equal values $h_r(S)$. Hence these values take a
common value $\sigma$ throughout $T$, which is a positive integer
because $\sigma|T|=n\beta$. It follows that $S$ is
$(\sigma,|T|)$-cohesive with $\sigma|T|=n\beta$, so condition~(ii)
applies. In every case, some voter in $S$ has satisfaction at least
$\mu_S(T)$, as required.
\end{proof}

Both boundary conditions are needed to retain the earlier guarantees.
To see that condition~(ii) is needed, consider two voters and two
rounds, where the first voter approves only $a$ and the second only
$b$ in both rounds. The outcome $(b,b)$ meets every requirement
imposed by conditions~(i) and~(iii), yet leaves the first voter with
nothing, although her proportional share is one round. This outcome
violates both FJR and EJR+, and condition~(ii) excludes it.

To see that condition~(iii) is needed, consider two voters and three
rounds, in which each voter approves exactly one candidate per round:
the first approves $b$, $a$, $b$ and the second approves $b$, $b$,
$a$ in rounds $1$, $2$, $3$. The outcome $(a,a,a)$ gives both voters
satisfaction one, whereas $(b,b,b)$ gives both satisfaction two;
hence $(a,a,a)$ violates FJR for $S=N$ and $T=[3]$. Nevertheless,
$(a,a,a)$ meets every requirement imposed by conditions~(i) and~(ii).
Each singleton $S$ has $\sum_r h_r(S)=3<2\cdot2$, while the whole
electorate has $(h_1(N),h_2(N),h_3(N))=(2,1,1)$, whose total
$4=2\cdot2$ is not exceeded, and whose largest cohesiveness product
$\sigma\tau$ is $3$. Condition~(iii) covers this boundary case even
though the agreement differs across rounds.

FJR+ can also demand representation for a group to which neither FJR
nor EJR+ guarantees anything. A group may have many approvals in total
without being able to distribute them so that every member benefits
within every proportional-size subset of rounds, and it may have no
candidate approved by all of its members. The following example
exhibits both phenomena.

\begin{example}
\label[example]{ex:strict}
There are five voters, three rounds, and three candidates. The table
lists the voters approving each candidate in each round; every voter
approves exactly one candidate per round.
\begin{center}
\begin{tabular}{cccc}
\toprule
Round & $a$ & $b$ & $c$\\
\midrule
1 & $\{1,2\}$ & $\{3,4\}$ & $\{5\}$\\
2 & $\{1,3\}$ & $\{2,4\}$ & $\{5\}$\\
3 & $\{1,4\}$ & $\{2,3\}$ & $\{5\}$\\
\bottomrule
\end{tabular}
\end{center}
The outcome $(c,c,c)$ gives satisfactions $(0,0,0,0,3)$. For
$S=\{1,2,3,4\}$, we have $h_r(S)=2$ in every round. Hence
$\sum_r h_r(S)=6>5$, and condition~(i) of FJR+ requires some member
of $S$ to have positive satisfaction. The outcome therefore violates
FJR+.

Nevertheless, $(c,c,c)$ satisfies FJR. A violating group cannot
contain voter~5, whose satisfaction already equals the number of
rounds, so consider groups among the other four voters. A singleton
has a proportional share of zero rounds in every set of rounds. A pair
is entitled to one round only if it specifies all three rounds, but
each pair agrees in only one of them; taking either of the other two
rounds as the one-round subset shows that the pair cannot guarantee
positive satisfaction to both members. A group of three is entitled to
at most one round, in which at most two of its members can be
satisfied. Finally, the group of all four voters cannot all be
satisfied within one round, and it is entitled to two rounds only if it
specifies all three. Two approving pairs from different rounds share
exactly one voter, because no pair is repeated in the table and the
two pairs in each round are disjoint. Hence two selected candidates
from different rounds satisfy at most three of the four voters, and
selecting $c$ does not help this group. Thus no group can guarantee
positive satisfaction to all of its members, as FJR would require.

The outcome also satisfies EJR+. A group containing voter~5 already
meets every possible satisfaction target. Among the other voters,
groups of three or four have no round in which all members approve a
common candidate, so EJR+ imposes no requirement on them. It remains to
consider singletons and pairs. A singleton is $(1,3)$-cohesive, while
a pair has round counts $(2,1,1)$ in some order. For either kind of
group, every product $\sigma\tau$ is at most $3<5$, so EJR+ requires
only satisfaction zero. We conclude that FJR+ is strictly stronger than
FJR and EJR+ combined.

The separating outcome is moreover Pareto efficient: preserving
voter~5's satisfaction requires selecting $c$ in every round, which
fixes the outcome.
\end{example}

\subsection{Party-List Preferences}

On party-list preferences, the agreement within every group is the same
in each round. In this case, the three conditions of FJR+ reduce to the
familiar lower-quota requirement.

\begin{proposition}
\label[proposition]{prop:party}
Suppose $N$ is partitioned into nonempty parties $N_1,\ldots,N_k$,
each with a distinct candidate, and every voter approves exactly her
party's candidate in every round. If $w_j$ is the number of rounds
won by party $j$, then an outcome satisfies FJR+ if and only if $w_j\geq\lfloor \ell|N_j| / n\rfloor$ for every $j \in [k]$.
\end{proposition}

\begin{proof}
Each party $N_j$ is $(|N_j|,\ell)$-cohesive, so condition~(ii) of FJR+
implies the stated lower quota. Conversely, suppose that every party
meets its lower quota, and fix a nonempty group $S$. Choose a party
$N_j$ maximizing $|S\cap N_j|$; since $S$ is nonempty and the parties
partition $N$, the set $S\cap N_j$ is nonempty. Then
$h_r(S)=|S\cap N_j|$ in every round, so every level $\beta$ required by
any of the three conditions is at most
$\lfloor\ell|S\cap N_j|/n\rfloor$. Every voter in $S\cap N_j$ has
satisfaction $w_j\geq\lfloor\ell|N_j|/n\rfloor$, which is at least
this level.
\end{proof}

\section{Polynomial-Time Verification}
\label{sec:verification}

A direct check of FJR+ would examine exponentially many groups. We
avoid this by following earlier verification algorithms that organize
voters by satisfaction thresholds: for EJR in approval-based
apportionment, where a candidate may receive several seats
\citep{brill2024approval}, for multiwinner EJR+
\citep{brill2023robust}, and for temporal EJR+
\citep{phillips2025strengthening}. For temporal FJR+, it suffices to
collect all voters below a satisfaction level, without fixing a
commonly approved candidate.

For each $\beta\in[\ell]$, let
$U_\beta=\{i\in N:\mathrm{sat}_i(\mathbf{o})<\beta\}$ be the set of
voters whose satisfaction is below $\beta$. Every group $S$ whose
members all have satisfaction below $\beta$ is a subset of $U_\beta$.
Enlarging a group cannot decrease $h_r$ in any round, and therefore
preserves $(\sigma,\tau)$-cohesiveness. Moreover, if a violation uses
condition~(iii), then $S=N$, and $S\subseteq U_\beta\subseteq N$
forces $U_\beta=N$. Consequently, if any group violates FJR+ at level
$\beta$, then so does $U_\beta$.

\begin{theorem}
\label{thm:verification}
An outcome satisfies FJR+ if and only if, for every $\beta\in[\ell]$, $\sum_{r=1}^{\ell}h_r(U_\beta)\leq n\beta$, where equality is allowed only if $U_\beta\neq N$ and the positive
values among $h_1(U_\beta),\ldots,h_\ell(U_\beta)$ are not all equal.
This condition can be checked in $O((n+\ell)m\ell)$ time.
\end{theorem}

\begin{proof}
Fix $\beta\in[\ell]$. If $U_\beta=\varnothing$, then
$\sum_r h_r(U_\beta)=0<n\beta$, and no group violates FJR+ at this
level. Otherwise, all members of $U_\beta$ have satisfaction below
$\beta$, so by the preceding argument, a violation exists at this level
if and only if one of the three conditions of \cref{def:fjrplus}
applies to $U_\beta$. Condition~(i) rules out a total exceeding
$n\beta$, and condition~(iii) rules out equality when $U_\beta=N$.

Now suppose that the total is at most $n\beta$. A violation under
condition~(ii) requires $\tau$ rounds whose values are at least
$\sigma$, with $\sigma\tau\geq n\beta$. This is possible only if the
total is exactly $n\beta$, those $\tau$ values all equal $\sigma$, and
every other value is zero. Conversely, if the total is $n\beta$ and
the positive values are all equal, they provide exactly such a choice
of $\sigma$ and $\tau$; note that there is at least one positive value
because the total is positive.

For the running time, first compute all voter satisfactions. For each
round $r$ and candidate $c$, count how many voters approving $c$ in
round $r$ have each satisfaction value from $0$ to $\ell$; including
initialization, this takes $O(n+\ell)$ time per round-candidate pair.
Accumulating these counts in increasing order of satisfaction gives,
for every level $\beta\in[\ell]$, the number of these voters in
$U_\beta$. Taking maxima over candidates then yields all values
$h_r(U_\beta)$ in $O((n+\ell)m\ell)$ time. Finally, for each $\beta$,
we add the $\ell$ values and compare the positive ones; these checks
take $O(\ell^2)$ time in total, within the stated bound.
\end{proof}

The verifier thus needs only satisfaction counts and approval counts.
For instance, suppose that $n\beta=4$ and $U_\beta\neq N$. Round
counts $(2,2,0,0)$ trigger condition~(ii), so the outcome violates
FJR+, whereas round counts $(2,1,1,0)$ pass the test at this level.
Both totals equal four, but only the first has equal positive values;
the zero entries correspond to rounds offering the group no approved
candidate and play no role in this distinction. When the test fails,
it also identifies the underrepresented group $U_\beta$ and the level
it deserves. The test decides the axiom exactly and applies to any
outcome, regardless of the voting rule that produced it.

\section{Polynomial-Time Computation}
\label{sec:computation}

We now show how to compute an FJR+ outcome in polynomial time by
modifying the $\varepsilon$-local search PAV rule
($\varepsilon$-lsPAV). Our rule builds on the polynomial-time local
search PAV rule of \citet{aziz2018complexity}, its temporal adaptation
by \citet{chandak2026proportional}, and the EJR+ guarantee for
temporal $\varepsilon$-lsPAV established by
\citet{phillips2025strengthening}. The modification subtracts a small
multiple of the minimum satisfaction from the PAV score.

To define the score, let $H_0=0$ and $H_j=\sum_{k=1}^j1/k$ for
$j\geq1$. The PAV score of an outcome is
$\Phi(\mathbf{o})=\sum_{i\in N}H_{\mathrm{sat}_i(\mathbf{o})}$.
When a voter's satisfaction rises from $j$ to $j+1$, the score
increases by $1/(j+1)$, so an improvement for a less satisfied voter
counts for more. Temporal local search PAV repeatedly changes the
candidate in a single round whenever doing so increases $\Phi$. Its
$\varepsilon$-lsPAV version accepts a change only if it increases
$\Phi$ by more than a positive threshold $\varepsilon$; a suitable
threshold yields both polynomial running time and the EJR+ guarantee.
We call an outcome a \emph{local optimum} of $\Phi$ if no single-round
change increases $\Phi$. In multiwinner voting, PAV can violate FJR
\citep{peters2021proportional}. The following example shows that, in
temporal voting, a local optimum of $\Phi$ can likewise violate FJR.

\begin{example}
\label[example]{ex:pav}
There are three voters, three rounds, and two candidates $a$ and $b$.
In round $r$, voter $r$ approves only $a$, and the other two voters
approve only $b$. The outcome $(a,a,a)$ gives every voter satisfaction
one. Changing any single round to $b$ changes the PAV score by
$2\cdot\tfrac12-1=0$, so no single-round change increases $\Phi$.
However, $(b,b,b)$ gives every voter satisfaction two. Taking $S=N$
and $T=[3]$ in the definition of FJR shows that $(a,a,a)$ violates
FJR.
\end{example}

The same tie in PAV scores occurs at every positive satisfaction
level, as the family of examples in \cref{prop:pav-family} shows.
The example suggests that insisting on a strict increase in $\Phi$ is
too restrictive: a change that leaves $\Phi$ unchanged can be a
necessary first step towards better representation. Our rule accepts
such a change when it \emph{lowers} the minimum satisfaction. This
direction may seem counterintuitive, but it is exactly what resolves
the equality cases in our proof.

The idea is as follows. Given a group violating FJR+, we compare the
current choice in each round with a candidate approved by as many
group members as possible. If none of these replacements increases
$\Phi$ substantially, we show that their changes in $\Phi$ sum to
zero, that none of them raises the minimum satisfaction, and that at
least one of them lowers it. Subtracting a small multiple of the
minimum satisfaction therefore makes at least one replacement increase
the modified score. The proof gives an explicit lower bound on this
increase, which ensures polynomial running time without requiring the
algorithm to identify a violating group.

Formally, define the modified score $\Psi$ and the improvement
threshold $\varepsilon$ by
\begin{equation}
\Psi(\mathbf{o})
=\Phi(\mathbf{o})-\delta\min_{i\in N}\mathrm{sat}_i(\mathbf{o}),
\quad
\delta=\frac{1}{4\ell^3},
\quad
\varepsilon=\frac{1}{8\ell^4}.
\label{eq:potential}
\end{equation}
For a round $r$ and candidate $c$, let $\mathbf{o}^{r\leftarrow c}$
denote the outcome obtained from $\mathbf{o}$ by replacing $o_r$ with
$c$. We fix an order of the candidates, and inspect rounds in
increasing order and, within each round, candidates in this fixed
order. The rule is as follows.

\begin{center}
\begin{minipage}{0.92\linewidth}
\hrule
\vspace{6pt}
\textbf{Modified $\varepsilon$-lsPAV for FJR+}
\begin{enumerate}[leftmargin=*,itemsep=3pt]
\item Start with the outcome $\mathbf{o}$ selecting the first candidate
in every round.
\item While some $r\in[\ell]$ and $c\in C$ satisfy
$\Psi(\mathbf{o}^{r\leftarrow c})-\Psi(\mathbf{o})>\varepsilon$,
use the first such pair $(r,c)$ to replace $\mathbf{o}$ with
$\mathbf{o}^{r\leftarrow c}$.
\item Return $\mathbf{o}$.
\end{enumerate}
\vspace{2pt}
\hrule
\end{minipage}
\end{center}

In \cref{ex:pav}, starting from $(a,a,a)$, the first change leaves
$\Phi$ unchanged and lowers the minimum satisfaction from one to zero,
so it increases $\Psi$ by $\delta>\varepsilon$. The fixed orders serve
only to make the rule deterministic: our guarantees hold for every
initial outcome and every choice among the changes exceeding the
threshold. We also stress that the negative term is a device for
obtaining the representation guarantee; it does not express a
preference about the distribution of satisfaction in the final
outcome.

\begin{theorem}
\label{thm:computation}
Every outcome returned by modified $\varepsilon$-lsPAV satisfies FJR+.
The rule makes fewer than $8n\ell^4H_\ell$ changes and can be
implemented in $O(n^2m\ell^5\log(\ell+1))$ time.
\end{theorem}

\subsection{The Rule Satisfies FJR+}

We prove that every FJR+ violation admits a single-round change that
increases $\Psi$ by at least twice the stopping threshold. The key
calculation bounds the loss in $\Phi$ caused by removing the current
choices. Summed over all rounds, this loss is at most $n$: a voter
with positive satisfaction $u$ contributes $1/u$ in each of the $u$
rounds in which she approves the selected candidate, and hence one in
total. Comparing this loss with the score contributed by alternative
candidates yields either a large increase in $\Phi$ or the equality
case described above.

\begin{lemma}
\label[lemma]{lem:improving-replacement}
If an outcome $\mathbf{o}$ violates FJR+, then some round $r\in[\ell]$ and candidate $c\in C$ satisfy $\Psi(\mathbf{o}^{r\leftarrow c})-\Psi(\mathbf{o}) \geq\frac{1}{4\ell^4}=2\varepsilon$.
\end{lemma}

\begin{proof}
Write $u_i=\mathrm{sat}_i(\mathbf{o})$. By \cref{thm:verification},
there is a level $\beta\in[\ell]$ such that $U=U_\beta$ is nonempty and
\begin{equation}
\sum_{r=1}^{\ell}h_r(U)\geq n\beta,
\label{eq:violating-total}
\end{equation}
where, if equality holds, either $U=N$ or all positive values
$h_r(U)$ are equal. In each round $r$, choose a candidate $c_r$
attaining $h_r(U)$; we call it the \emph{comparison candidate} of
round $r$. All replacements below are evaluated at the fixed outcome
$\mathbf{o}$, not applied one after another. Write
$G_{r,c_r}=\Phi(\mathbf{o}^{r\leftarrow c_r})-\Phi(\mathbf{o})$ for
the resulting change in the PAV score. Let
$W_r=\{i\in N:o_r\in a_{i,r}\}$ and
$V_r=\{i\in N:c_r\in a_{i,r}\}$ be the sets of voters approving the
current and the comparison candidate, respectively. By the choice of
$c_r$, we have $|V_r\cap U|=h_r(U)$. Every voter in $W_r$ has positive
satisfaction, and the score lost by removing the current choice in
round $r$ is $p_r=\sum_{i\in W_r}\frac{1}{u_i}$. Since each voter with
$u_i>0$ contributes $1/u_i$ in exactly $u_i$ rounds,
\begin{equation}
\sum_{r=1}^{\ell}p_r=|\{i\in N:u_i>0\}|\leq n.
\label{eq:total-loss}
\end{equation}
This calculation of the loss from removing selected candidates is a
standard ingredient in the analysis of PAV in multiwinner voting
\citep{aziz2017justified}, public decisions
\citep{skowron2022public}, and temporal voting
\citep{chandak2026proportional,phillips2025strengthening}.

\emph{Comparing approval counts with changes in the PAV score.}
Replacing the current candidate costs $p_r$. The comparison candidate
then contributes $1/(u_i+1)$ for each voter in $V_r\setminus W_r$ and
restores $1/u_i$ for each voter in $V_r\cap W_r$. Thus
\begin{equation}
G_{r,c_r}
=\sum_{i\in V_r\setminus W_r}\frac{1}{u_i+1}
  +\sum_{i\in V_r\cap W_r}\frac{1}{u_i}-p_r =\frac{h_r(U)}{\beta}-p_r+e_r,
\label{eq:gain}
\end{equation}
where $e_r$ is the amount by which the score contributed by the
comparison candidate exceeds $h_r(U)/\beta$:
\begin{equation}
e_r=
\sum_{i\in V_r\setminus W_r}\frac{1}{u_i+1}
+\sum_{i\in V_r\cap W_r}\frac{1}{u_i}
-\frac{|V_r\cap U|}{\beta}.
\label{eq:extra-gain}
\end{equation}
Each voter's contribution to $e_r$ is nonnegative. An approver in
$U\setminus W_r$ contributes $1/(u_i+1)-1/\beta$, which is nonnegative
because $u_i<\beta$; an approver in $U\cap W_r$ contributes
$1/u_i-1/\beta>0$; and an approver outside $U$ contributes a positive
reciprocal.

We next show that a positive $e_r$ cannot be arbitrarily small. Every
denominator in these contributions lies in $[\ell]$; in particular, if
$i\notin W_r$, then $u_i\leq\ell-1$. For integers
$1\leq a<b\leq\ell$, we have $1/a-1/b=(b-a)/(ab)\geq1/\ell^2$, and a
positive reciprocal is at least $1/\ell\geq1/\ell^2$. It follows that
\begin{equation}
e_r=0\quad\text{or}\quad e_r\geq\frac{1}{\ell^2}.
\label{eq:positive-gain}
\end{equation}
Moreover, since only an approver in $U\setminus W_r$ with
satisfaction exactly $\beta-1$ contributes zero, $e_r=0$ implies
\begin{equation}
V_r\subseteq\{i\in N:u_i=\beta-1\}\subseteq U,
\quad V_r\cap W_r=\varnothing.
\label{eq:approvers}
\end{equation}

\emph{Either a large increase in the PAV score or exact equality.}
A single-round replacement changes the minimum satisfaction by at most
one. Hence, if $G_{r,c_r}\geq1/\ell^3$ for some round $r$, then
\[
\Psi(\mathbf{o}^{r\leftarrow c_r})-\Psi(\mathbf{o})
\geq\frac{1}{\ell^3}-\delta
=\frac{3}{4\ell^3}
\geq\frac{1}{4\ell^4},
\]
and the lemma holds. We may therefore assume that
\begin{equation}
G_{r,c_r}<\frac{1}{\ell^3}
\quad\text{for every }r\in[\ell].
\label{eq:small-gains}
\end{equation}
Summing \eqref{eq:gain} over all rounds and using
\eqref{eq:total-loss} gives
\begin{equation}
\sum_{r=1}^{\ell}G_{r,c_r}
=\frac{\sum_r h_r(U)-n\beta}{\beta}
 +n-|\{i:u_i>0\}|+\sum_{r=1}^{\ell}e_r
<\frac{1}{\ell^2}.
\label{eq:equality}
\end{equation}
All terms to the right of the equality sign are nonnegative, and each
of them is either zero or at least $1/\ell^2$. Indeed, if the first
term is positive, integrality makes it at least
$1/\beta\geq1/\ell\geq1/\ell^2$; if the second term is positive, it
is at least one; and a positive $e_r$ is at least $1/\ell^2$ by
\eqref{eq:positive-gain}. Hence all these terms are zero. In
particular,
\begin{equation}
\sum_{r=1}^{\ell}h_r(U)=n\beta,
\quad u_i>0\text{ for every }i\in N,
\quad e_r=0\text{ for every }r\in[\ell],
\label{eq:exact-equality}
\end{equation}
and
\begin{equation}
\sum_{r=1}^{\ell}G_{r,c_r}=0.
\label{eq:zero-total}
\end{equation}
Consequently, \eqref{eq:approvers} holds in every round. Moreover,
$\beta\geq2$, since $U$ is nonempty and every voter has positive
satisfaction. Note that the lower bounds on positive values apply only
to the nonnegative terms in \eqref{eq:equality}; the individual changes
$G_{r,c_r}$ may have either sign.

\emph{No comparison candidate is approved by all of $U$.}
Since the total in \eqref{eq:violating-total} is now exactly
$n\beta$, our choice of $\beta$ leaves two cases. We show that in both
cases, every $V_r\subseteq U$ omits at least one voter of $U$.

First suppose that all positive values $h_r(U)$ are equal, and
suppose for contradiction that one of them equals $|U|$; then all
positive values do. Choose a voter in $U$ and a round $r$ in which
she approves the current candidate; such a round exists by
\eqref{eq:exact-equality}. The current candidate then has an approver
in $U$ in round $r$, so $h_r(U)>0$ and hence $|V_r\cap U|=|U|$. The
chosen voter therefore belongs to $V_r\cap W_r$, contradicting
\eqref{eq:approvers}. Thus $h_r(U)<|U|$ in every round. Since
\eqref{eq:approvers} gives $V_r\subseteq U$ and $|V_r|=h_r(U)$, we
obtain $V_r\subsetneq U$ in every round.

Otherwise, $U=N$. If $V_r=N$ in some round $r$, then
\eqref{eq:approvers} gives $W_r=\varnothing$ and $u_i=\beta-1$ for
all voters. Hence $G_{r,c_r}=n/\beta\geq1/\ell\geq1/\ell^3$,
contradicting \eqref{eq:small-gains}. Again, $V_r\subsetneq U$ in
every round.

\emph{The effect on the minimum satisfaction.}
Let $u_{\min}=\min_i u_i$, and let $d_r=\min_i\mathrm{sat}_i(\mathbf{o}^{r\leftarrow c_r})-u_{\min}$ be the change in the minimum satisfaction caused by replacing $o_r$
with $c_r$. We first show that no replacement raises the minimum.
Since $U$ is nonempty, $u_{\min}\leq\beta-1$. If
$u_{\min}<\beta-1$, then by \eqref{eq:approvers}, no voter at the
minimum approves $c_r$, so $d_r\leq0$. If $u_{\min}=\beta-1$, then
every voter in $U$ is at the minimum, and at least one of them is not
in $V_r$ because $V_r\subsetneq U$. This voter gains no approval, so
again $d_r\leq0$.

We next show that at least one replacement lowers the minimum. Choose
a voter at satisfaction $u_{\min}$. Her satisfaction is positive by
\eqref{eq:exact-equality}, so she approves the current candidate in
some round $r$. Since $V_r$ and $W_r$ are disjoint, she loses one
approval in that replacement. No voter loses more than one approval,
so $d_r=-1$ in this round. Therefore $\sum_r d_r\leq-1$, and together
with \eqref{eq:zero-total} we obtain
\[
\sum_{r=1}^{\ell}
(\Psi(\mathbf{o}^{r\leftarrow c_r})-\Psi(\mathbf{o}))
=\sum_{r=1}^{\ell}G_{r,c_r}-\delta\sum_{r=1}^{\ell}d_r
\geq\delta.
\]
Hence at least one of the $\ell$ summands is at least
$\delta/\ell=1/(4\ell^4)$, as required.
\end{proof}

\begin{proof}[Proof of the FJR+ guarantee in \cref{thm:computation}]
At a returned outcome, every single-round replacement increases $\Psi$
by at most $\varepsilon$. If this outcome violated FJR+,
\cref{lem:improving-replacement} would provide a replacement
increasing $\Psi$ by at least $2\varepsilon>\varepsilon$, a
contradiction.
\end{proof}

\subsection{Running Time}

\begin{proof}[Proof of the running-time bound in \cref{thm:computation}]
We first show that every outcome $\mathbf{o}$ satisfies
\[
0\leq\Psi(\mathbf{o})\leq nH_\ell.
\]
The upper bound follows from \eqref{eq:potential}. For the lower
bound, if some voter has satisfaction zero, the subtracted term
vanishes. Otherwise, $\Phi(\mathbf{o})\geq n\geq1$, whereas
$\delta\min_i\mathrm{sat}_i(\mathbf{o})\leq\delta\ell\leq1/4$.
Since every accepted change increases $\Psi$ by more than
$\varepsilon$, there are fewer than
$nH_\ell/\varepsilon=8n\ell^4H_\ell$ changes. If the current
satisfactions are maintained, a single replacement and its effect on
the minimum satisfaction can be evaluated in $O(n)$ time, so scanning
all $m\ell$ possible replacements takes $O(nm\ell)$ time. Including
the final scan, this gives a total running time of
$O(n^2m\ell^5\log(\ell+1))$.
\end{proof}

Neither algorithm nor their proofs rely on all rounds sharing the same
candidate set. If each round $r$ has its own nonempty candidate set
$C_r$, it suffices to replace $C$ by $C_r$ wherever candidates in that
round are considered. Writing $M=\sum_r|C_r|$, verification then takes
$O((n+\ell)M)$ time and computation takes $O(n^2M\ell^4\log(\ell+1))$
time.

Our computation guarantee does not include Pareto efficiency. For
example, consider three voters whose unique approved candidates in the
three rounds are $(a,b,a)$, $(a,b,b)$, and $(b,a,c)$, respectively,
with $a$ first in the candidate order. The initial outcome $(a,a,a)$
gives satisfactions $(2,1,1)$, and every single-round replacement
changes $\Psi$ by a nonpositive amount, so the rule returns this
outcome. Yet $(a,b,c)$ gives satisfactions $(2,2,1)$ and is therefore
a Pareto improvement. In fact, every Pareto improvement increases
$\Psi$: each unit of satisfaction gained raises $\Phi$ by at least
$1/\ell$, the minimum satisfaction rises by at most the total
satisfaction gained, and $\delta<1/\ell$. The obstacle is that a
Pareto improvement may require changing several rounds at once, as in
this example, whereas our rule considers only single-round changes.

\subsection{Representation at Integer Levels}
\label{sec:boundary}

We now return to the stronger requirement \eqref{eq:stronger-total},
which replaces $>$ by $\geq$ in condition~(i) and thereby implies all
three conditions of FJR+. The same verification argument applies:
\eqref{eq:stronger-total} is equivalent to
$\sum_r h_r(U_\beta)<n\beta$ for every $\beta\in[\ell]$. Moreover,
this stronger requirement can always be satisfied; what separates it
from our main result is computation.

\begin{proposition}
\label[proposition]{prop:integer-boundary}
Suppose an outcome $\mathbf{o}$ admits no single-round replacement
that increases $\Phi$, and no such replacement that preserves $\Phi$
and lowers minimum satisfaction. Then $\mathbf{o}$ satisfies
\eqref{eq:stronger-total}. Such an outcome exists in every temporal
election.
\end{proposition}

\begin{proof}
Suppose for contradiction that $\sum_r h_r(U_\beta)\geq n\beta$ for
some $\beta\in[\ell]$; in particular, $U_\beta$ is nonempty. We use
the notation and comparison candidates from the proof of
\cref{lem:improving-replacement}, with $U=U_\beta$. The identities
\eqref{eq:gain} and~\eqref{eq:total-loss}, as well as the
nonnegativity of each $e_r$, hold at any outcome. Since no replacement
increases $\Phi$, every $G_{r,c_r}$ is nonpositive, and therefore
\[
0\geq\sum_r G_{r,c_r}
=\frac{\sum_r h_r(U)-n\beta}{\beta}
 +n-|\{i:u_i>0\}|+\sum_r e_r
\geq0.
\]
As each of the three terms in the middle is nonnegative, equality
implies that every voter has positive satisfaction and every $e_r$ is
zero. Moreover, every $G_{r,c_r}$ is zero, since these changes are
nonpositive and sum to zero. Now choose a voter of minimum satisfaction
and a round in which she approves the selected candidate. By
\eqref{eq:approvers}, she does not approve the comparison candidate in
that round. The corresponding replacement therefore lowers the minimum
satisfaction by one while leaving $\Phi$ unchanged, contrary to the
hypothesis.

For existence, the finite set $C^\ell$ contains an outcome maximizing
$\Phi$. Among all such outcomes, choose one with the smallest minimum
satisfaction. A replacement increasing $\Phi$ would contradict the
first choice, and a replacement preserving $\Phi$ while lowering the
minimum would contradict the second. Hence this outcome satisfies both
local conditions. It is also Pareto efficient, since a Pareto
improvement would strictly increase $\Phi$.
\end{proof}

Why does this argument not yield a polynomial-time algorithm for
\eqref{eq:stronger-total}? \Cref{prop:integer-boundary} requires an
outcome at which no replacement increases $\Phi$ and none preserves
$\Phi$ while lowering the minimum satisfaction. Our rule, by contrast,
stops as soon as no replacement increases $\Psi$ by more than
$\varepsilon$. In the equality case of
\cref{lem:improving-replacement}, the changes in $\Phi$ need only sum
to zero, and individual changes may be positive or negative. The two
boundary conditions of FJR+ ensure that none of the comparison
replacements raises the minimum satisfaction, so their changes in
$\Psi$ sum to at least $\delta$; dividing by the number of rounds gives
the required improvement. Without these conditions, the argument cannot
exclude replacements that raise the minimum. Thus, although existence
and polynomial-time verification hold for both requirements, we can
bound the number of changes only for FJR+. Our polynomial-time
guarantee covers all integer levels required by FJR and EJR+, while
whether \eqref{eq:stronger-total} can be attained in polynomial time
remains open.

In multiwinner voting, \citet{kraiczy2024lower} show that local PAV can
take superpolynomially many steps under a fixed order of replacements
when every positive improvement is accepted. Their result, however,
does not settle the complexity of \eqref{eq:stronger-total} in
temporal voting.

\section{Discussion}
\label{sec:discussion}

FJR+ shows that flexible group representation in temporal voting is
compatible with efficient computation and verification. The two
algorithms play different roles: modified $\varepsilon$-lsPAV produces
an outcome, whereas the verifier checks any proposed outcome,
regardless of how it was obtained. This separation leaves room for
other voting rules that satisfy the same axiom.

Several questions remain. Our polynomial bound on the number of
changes is large, and a faster rule satisfying FJR+ would make the
guarantee more useful in applications. Relatedly, it would be useful
to understand how the initial outcome and the choice among improving
replacements affect the number of changes and the resulting voter
satisfaction in practice. We also do not know whether a
Pareto-efficient FJR+ outcome can be found in polynomial time. The comparison in \cref{sec:multiwinner} explains why
multiwinner algorithms do not directly settle these questions: although
the two FJR+ axioms coincide under static preferences when each round
is represented by separate candidates, their requirements differ once
agreement changes across rounds.

Moreover, one could seek a suitable group proportionality axiom
for temporal voting over public chores \citep{elkind2025temporalchores}, potentially consider a stronger proportionality axiom in this setting (e.g., adapt the recently discovered \emph{core}+ notion defined for multiwinner voting \citep{becker2026existencecoreapprovalbasedcommittee}), or quantify the loss in total voter satisfaction required by FJR+, extending the study of weaker temporal representation axioms by \citet{teh2026price}.

\subsection*{Declaration of generative AI use}
The author used GPT-6 Astra as a research assistant when exploring proof ideas and revising the exposition. He wrote the mathematical arguments and takes full responsibility for the content and conclusions of the work.

\bibliographystyle{plainnat}
\bibliography{bib}

@inproceedings{peters2021proportional,
  author = {Dominik Peters and Grzegorz Pierczy{\'n}ski and Piotr Skowron},
  title = {Proportional Participatory Budgeting with Additive Utilities},
  booktitle = {Proceedings of the 35th International Conference on Neural Information Processing Systems},
  pages = {12726--12737},
  year = {2021},
}

@inproceedings{brill2023robust,
  author = {Markus Brill and Jannik Peters},
  title = {Robust and Verifiable Proportionality Axioms for Multiwinner Voting},
  booktitle = {Proceedings of the 24th ACM Conference on Economics and Computation},
  pages = {301},
  year = {2023},
}

@article{bulteau2021justified,
  author = {Laurent Bulteau and Noam Hazon and Rutvik Page and Ariel Rosenfeld and Nimrod Talmon},
  title = {Justified Representation for Perpetual Voting},
  journal = {IEEE Access},
  volume = {9},
  pages = {96598--96612},
  year = {2021}
}

@article{chandak2026proportional,
  author = {Nikhil Chandak and Shashwat Goel and Dominik Peters},
  title = {Proportional Aggregation of Preferences for Sequential Decision Making},
  journal = {Journal of Artificial Intelligence Research},
  volume = {85},
  year = {2026},
}

@inproceedings{phillips2025strengthening,
  author = {Bradley Phillips and Edith Elkind and Nicholas Teh
            and Tomasz W{\k{a}}s},
  title = {Strengthening Proportionality in Temporal Voting},
  booktitle = {Proceedings of the 25th International Conference on Autonomous Agents and Multiagent Systems},
  pages = {3823--3832},
  year = {2026},
  note = {Full version: arXiv:2505.22513},
}

@article{frank2026polynomial,
  author = {Fabian Frank and Jannik Peters},
  title = {A Polynomial-Time Rule Satisfying Full Justified Representation},
  year = {2026},
  journal = {arXiv preprint arXiv:2608.05397},
}

@article{neoh2026closinggaps,
  author = {Tzeh Yuan Neoh and Nicholas Teh},
  title = {Closing Gaps in Online Fair Division},
  year = {2026},
  journal = {arXiv preprint arXiv:2609.05310},
}

@article{ai2026full,
  author = {Yizhou Ai},
  title = {Full Justified Representation under {Hare} and {Droop} Quotas in Polynomial Time},
  year = {2026},
  journal = {arXiv preprint arXiv:2608.05417},
}

@article{teh2026strengthening,
  author = {Nicholas Teh},
  title = {Strengthening Full Justified Representation: Efficient Verification and Computation},
  year = {2026},
  journal = {arXiv preprint arXiv:2608.11500},
}

@article{neoh2026strengthening,
  author = {Tzeh Yuan Neoh and Nicholas Teh},
  title = {Strengthening Proportionality in Participatory Budgeting with Additive Utilities},
  year = {2026},
  journal = {arXiv preprint arXiv:2609.08888},
}

@article{aziz2017justified,
  author = {Haris Aziz and Markus Brill and Vincent Conitzer and
            Edith Elkind and Rupert Freeman and Toby Walsh},
  title = {Justified Representation in Approval-Based Committee Voting},
  journal = {Social Choice and Welfare},
  volume = {48},
  pages = {461--485},
  year = {2017},
}

@inproceedings{aziz2018complexity,
  author = {Haris Aziz and Edith Elkind and Shenwei Huang and
            Martin Lackner and Luis S{\'a}nchez-Fern{\'a}ndez and
            Piotr Skowron},
  title = {On the Complexity of Extended and Proportional Justified Representation},
  booktitle = {Proceedings of the 32nd AAAI Conference on Artificial Intelligence (AAAI)},
  pages = {902--909},
  year = {2018},
}

@article{brill2024approval,
  author = {Markus Brill and Paul G{\"o}lz and Dominik Peters and
            Ulrike Schmidt-Kraepelin and Kai Wilker},
  title = {Approval-Based Apportionment},
  journal = {Mathematical Programming},
  volume = {203},
  pages = {77--105},
  year = {2024},
}

@inproceedings{conitzer2017fair,
  author = {Vincent Conitzer and Rupert Freeman and Nisarg Shah},
  title = {Fair Public Decision Making},
  booktitle = {Proceedings of the 18th ACM Conference on Economics and Computation},
  pages = {629--646},
  year = {2017},
}

@inproceedings{elkind2025verifying,
  author = {Edith Elkind and Svetlana Obraztsova and Jannik Peters
            and Nicholas Teh},
  title = {Verifying Proportionality in Temporal Voting},
  booktitle = {Proceedings of the 39th AAAI Conference on Artificial Intelligence (AAAI)},
  pages = {13805--13813},
  year = {2025},
}

@article{golz2026approval,
  author = {Paul G{\"o}lz and Hannane Yaghoubizade},
  title = {Approval-Based Apportionment: Like Portioning,
           Approximately like Committee Voting},
  year = {2026},
  journal = {arXiv preprint arXiv:2608.27605},
}

@inproceedings{kraiczy2024lower,
  author = {Sonja Kraiczy and Edith Elkind},
  title = {A Lower Bound for Local Search Proportional Approval Voting},
  booktitle = {Proceedings of the 32nd Annual European Symposium on Algorithms (ESA)},
  pages = {82:1--82:14},
  year = {2024},
}

@inproceedings{lackner2020perpetual,
  author = {Martin Lackner},
  title = {Perpetual Voting: Fairness in Long-Term Decision Making},
  booktitle = {Proceedings of the 34th AAAI Conference on Artificial Intelligence (AAAI)},
  pages = {2103--2110},
  year = {2020},
}

@inproceedings{lackner2023proportional,
  author = {Martin Lackner and Jan Maly},
  title = {Proportional Decisions in Perpetual Voting},
  booktitle = {Proceedings of the 37th AAAI Conference on Artificial Intelligence (AAAI)},
  pages = {5722--5729},
  year = {2023},
}

@inproceedings{elkind2024temporalelections,
    author = {Edith Elkind and Tzeh Yuan Neoh and Nicholas Teh},
    title = {Temporal Elections: Welfare, Strategyproofness, and Proportionality},
    year = {2024},
    booktitle = {Proceedings of the 27th European Conference on Artificial Intelligence (ECAI)},
    pages = {3292--3299},
}

@inproceedings{elkind2022temporalslot,
  title     = {Fairness in Temporal Slot Assignment},
  author    = {Edith Elkind and Sonja Kraiczy and Nicholas Teh},
  booktitle = {Proceedings of the 15th International Symposium on Algorithmic Game Theory (SAGT)},
  pages     = {490--507},
  year      = {2022}
}

@inproceedings{elkind2025temporalchores,
    author = {Edith Elkind and Tzeh Yuan Neoh and Nicholas Teh},
    title = {Not in My Backyard! {T}emporal Voting Over Public Chores},
    year = {2025},
    booktitle = {Proceedings of the 34th International Joint Conference on Artificial Intelligence (IJCAI)},
    pages = {3814--3820}
}

@inproceedings{teh2026price,
    author = {Nicholas Teh},
    title = {The Price of Proportional Representation in Temporal Voting},
    year = {2026},
    pages = {3541--3549},
    booktitle = {Proceedings of the 35th International Joint Conference on Artificial Intelligence (IJCAI)}
}

@inproceedings{choo2026approxproponline,
  author    = {Davin Choo and Winston Fu and Derek Khu and Tzeh Yuan Neoh and Tze-Yang Poon and Nicholas Teh},
  title     = {Approximate Proportionality in Online Fair Division},
  booktitle = {Proceedings of the 43rd International Conference on Machine Learning (ICML)},
  year      = {2026},
  note      = {Extended version available as arXiv:2508.03253}
}

@inproceedings{neoh2026online,
  author    = {Tzeh Yuan Neoh and Jannik Peters and Nicholas Teh},
  title     = {Online Fair Division with Additional Information},
  booktitle = {Proceedings of the 43rd International Conference on Machine Learning (ICML)},
  year      = {2026},
  note      = {Extended version available as arXiv:2505.24503}
}

@inproceedings{elkind2024temporal,
  author       = {Edith Elkind and
                  Svetlana Obraztsova and
                  Nicholas Teh},
  title        = {Temporal Fairness in Multiwinner Voting},
booktitle = {Proceedings of the 38th AAAI Conference on Artificial Intelligence (AAAI)},
    year      = {2024},
    pages     = {22633--22640},
}

@inproceedings{zech2024multiwinnerchange,
    author = {Valentin Zech and Niclas Boehmer and Edith Elkind and Nicholas Teh},
    title = {Multiwinner Temporal Voting with Aversion to Change},
    year = {2024},
    booktitle = {Proceedings of the 27th European Conference on Artificial Intelligence (ECAI)},
    pages        = {3236--3243},
}

@article{cohen2026budgetconstraints,
  author  = {Saar Cohen and Nicholas Teh and Paul W. Goldberg and Michael J. Wooldridge},
  title   = {Online Fair Division with Budget Constraints},
  journal = {arXiv preprint arXiv:2607.23310},
  year    = {2026}
}

@inproceedings{Goldberg2026,
   author = {Paul Goldberg and Isaac Robinson and Nicholas Teh},
   booktitle = {Proceedings of the 19th International Symposium on Algorithmic Game Theory (SAGT)},
   title = {Minimizing Cumulative Envy in Allocating a Sequence of Items},
   year = {2026},
   note = {Forthcoming}
}

@inproceedings{alouf2022better,
  title     = {Better Collective Decisions via Uncertainty Reduction},
  author    = {Alouf-Heffetz, Shiri and Bulteau, Laurent and Elkind, Edith and Talmon, Nimrod and Teh, Nicholas},
  booktitle = {Proceedings of the 31st International Joint Conference on
               Artificial Intelligence (IJCAI)},
  pages     = {24--30},
  year      = {2022}
}

@inproceedings{elkind2025temporal,
  title={Temporal Fair Division of Indivisible Items},
  author={Elkind, Edith and Lam, Alexander and Latifian, Mohamad and Neoh, Tzeh Yuan and Teh, Nicholas},
  booktitle={Proceedings of the 24th International Conference on Autonomous Agents and Multiagent Systems (AAMAS)},
  pages={676--685},
  year={2025}
}

@article{choi2026tfdmm,
      author={Kui-Wang Choi and Minming Li and Nicholas Teh},
      title={Temporal Fair Division of Indivisible Mixed Manna: Tractable Settings}, 
    journal       = {arXiv preprint },
    year          = {2026},
    publisher     = {},
    volume        = {arXiv:2608.20033}
}

@article{becker2026existencecoreapprovalbasedcommittee,
      title={Existence of the Core in Approval-Based Committee Elections}, 
      author={Patrick Becker and Matthias Greger and Dominik Peters},
      year={2026},
    journal       = {arXiv preprint },
    volume        = {arXiv:2609.11912}
}

@inproceedings{masarik2024generalised,
  author = {Tom{\'a}{\v s} Masa{\v r}{\'i}k and
            Grzegorz Pierczy{\'n}ski and Piotr Skowron},
  title = {A Generalised Theory of Proportionality in Collective Decision Making},
  booktitle = {Proceedings of the 25th ACM Conference on Economics and Computation (EC)},
  pages = {734--754},
  year = {2024},
}

@inproceedings{skowron2022public,
  author = {Piotr Skowron and Adrian G{\'o}recki},
  title = {Proportional Public Decisions},
  booktitle = {Proceedings of the 36th AAAI Conference on Artificial Intelligence (AAAI)},
  pages = {5191--5198},
  year = {2022},
}

\clearpage
\appendix

\section{Harmonic Local Optima at Higher Satisfaction Levels}
\label{app:pav-family}

The following construction extends \cref{ex:pav} to every positive
satisfaction level. In each election of the family, no single-round
change increases the PAV score $\Phi$, yet the outcome violates FJR.

\begin{proposition}
\label[proposition]{prop:pav-family}
For every integer $k\geq1$, there is a temporal election with
$n=\ell=2k+1$ and two candidates $a,b$, in which each voter approves
exactly one candidate per round, with the following properties.
The outcome selecting $a$ in every round gives every voter
satisfaction $k$, and every replacement of $a$ by $b$ leaves
$\Phi$ unchanged. Nevertheless, it violates FJR, and each such
replacement increases the modified score $\Psi$ by $\delta$.
\end{proposition}

\begin{proof}
In round $r$, let the $k$ voters $r,r+1,\ldots,r+k-1$ approve $a$,
with indices taken cyclically, so that voter $1$ follows voter $2k+1$;
the other $k+1$ voters approve $b$. Each voter approves $a$ in exactly
$k$ rounds, so the outcome $\mathbf{a}=(a,\ldots,a)$ gives every voter
satisfaction $k$. Changing any round from $a$ to $b$ removes one
approval from each of its $k$ current approvers and adds one approval
for each of the other $k+1$ voters. The PAV score therefore changes by
exactly
\[
(k+1)\frac{1}{k+1}-k\frac{1}{k}=0,
\]
and no single-round replacement increases it.

Selecting $b$ in every round gives every voter satisfaction $k+1$. For
$S=N$ and $T=[\ell]$, the only subset $R$ considered in the definition
of $\mu_N(T)$ is $T$ itself, so $\mu_N(T)\geq k+1$. Conversely,
$h_r(N)=k+1$ in every round, so every outcome has total satisfaction
at most $n(k+1)$, and its minimum satisfaction cannot exceed $k+1$.
Hence $\mu_N(T)=k+1$, and the outcome $\mathbf{a}$ violates FJR.

Finally, each replacement at $\mathbf{a}$ leaves $k$ voters at
satisfaction $k-1$ and $k+1$ voters at satisfaction $k+1$. Since
$k\geq1$, the minimum satisfaction falls from $k$ to $k-1$ while the
PAV score is unchanged, so $\Psi$ increases by $\delta$, which exceeds
$\varepsilon$.
\end{proof}

\section{Comparison with Multiwinner Voting}
\label{sec:multiwinner}
Temporal FJR+ and the multiwinner FJR+ axiom of
\citet{teh2026strengthening} both strengthen familiar representation
guarantees, and both admit polynomial-time computation and
verification. These shared properties, however, do not make the
definitions interchangeable. To compare them, we first associate with
each temporal election a multiwinner election with exactly the same
approvals and voter satisfactions. Under this correspondence, neither
axiom implies the other in general, but the two coincide when
preferences are static.

\subsection{The Multiwinner Axiom}

Following the notation of \citet{phillips2025strengthening}, we mark
multiwinner axioms with the superscript $\mathit{mw}$. In a
\emph{multiwinner election}, each voter $i\in N$ submits an approval
set $a_i\subseteq C$, and a \emph{committee} is a set $W\subseteq C$
of a prescribed size $k\in[|C|]$. Voter $i$'s satisfaction is
$\mathrm{sat}_i(W)=|a_i\cap W|$.

FJR$^\textit{\,mw}$ \citep{peters2021proportional} requires, for every
group $S\subseteq N$ and candidate set $T\subseteq C$ with
$|T|\leq k|S|/n$, that
\[
\max_{i\in S}\mathrm{sat}_i(W)
\geq \min_{i\in S}|a_i\cap T|.
\]
Thus a group compares its representation with what all of its members
could obtain from a candidate set within its proportional share.
EJR$^\textit{\,mw}$+ \citep{brill2023robust} requires that, for every
$\beta\in[k]$ and every unelected candidate $c\in C\setminus W$,
\begin{equation}
|\{i\in N:c\in a_i,\ \mathrm{sat}_i(W)<\beta\}|
<\frac{n\beta}{k}.
\label{eq:mw-ejrplus}
\end{equation}
In other words, if sufficiently many voters approve the same unelected
candidate, they cannot all remain below the corresponding satisfaction
level.

FJR$^\textit{\,mw}$+ allows voters and comparison candidates to carry
nonnegative weights. A group's size is then its total voter weight,
the size of its comparison set is its total candidate weight, and each
voter asks for the target satisfaction level multiplied by her weight.
We state the definition of \citet{teh2026strengthening} in our
satisfaction notation.

\begin{definition}[FJR$^\textit{\,mw}$+]
\label{def:mw-fjrplus}
A committee $W$ satisfies FJR$^\textit{\,mw}$+ if there is no
$\beta\in[k]$ and no choice of nonnegative weights $z_i$ for voters
and $y_c$ for candidates, and nonnegative assigned amounts
$x_{ic}$, for $i\in N$ and $c\in C$, satisfying
\begin{align}
z_i&=0
&&\text{if }\mathrm{sat}_i(W)\geq\beta,
\nonumber\\
x_{ic}&=0
&&\text{if }c\notin a_i,
\nonumber\\
\sum_{c\in C}x_{ic}&=\beta z_i
&&\text{for every }i\in N,
\nonumber\\
x_{ic}&\leq y_c
&&\text{for every }i\in N,\ c\in C,
\nonumber\\
x_{iw}&\leq z_i
&&\text{for every }i\in N,\ w\in W,
\label{eq:mw-assignments}
\end{align}
together with
\begin{equation}
\sum_{i\in N}z_i>0,
\quad
\sum_{i\in N}z_i\geq\frac{n}{k}\sum_{c\in C}y_c.
\label{eq:mw-proportionality}
\end{equation}
\end{definition}

Here $x_{ic}$ is the amount of candidate $c$ assigned to voter $i$.
Only voters below the target $\beta$ may have positive weight, and each
such voter must receive $\beta$ units per unit of her weight from
candidates she approves. Each assignment is bounded by the candidate's
weight; this bound applies separately to each voter, just as an
elected candidate gives one unit of satisfaction to every voter who
approves it. The additional bound $x_{iw}\leq z_i$ applies only to
elected candidates, and limits the contribution of each approved
elected candidate to one unit of satisfaction per unit of voter
weight. Finally, \eqref{eq:mw-proportionality} requires the weighted
group to be large enough to justify the weighted comparison set. The
weights appear only in the definition: the selected committee remains
an ordinary set of $k$ candidates.

This axiom implies both FJR$^\textit{\,mw}$ and EJR$^\textit{\,mw}$+
\citep{teh2026strengthening}. In particular, an EJR$^\textit{\,mw}$+
violation by a group $S$ and an unelected candidate $c$ yields the
assignments $z_i=1$ and $x_{ic}=\beta$ for $i\in S$, with $y_c=\beta$
and all other variables zero. This implication relies on the absence
of an upper bound $x_{ic}\leq z_i$ for unelected candidates.

\subsection{Preserving Approvals and Satisfaction}

We use the candidate-round construction of
\citet{phillips2025strengthening}: given a temporal election
$E=(C,N,\ell,A)$, we make each candidate-round pair a distinct
multiwinner candidate. Formally, let
\begin{equation}
\widehat C=C\times[\ell],
\quad
\widehat a_i=\{(c,r):c\in a_{i,r}\},
\quad k=\ell.
\label{eq:mw-conversion}
\end{equation}
An outcome $\mathbf{o}$ then corresponds to the committee
\[
W_{\mathbf{o}}=\{(o_r,r):r\in[\ell]\}.
\]
This construction preserves satisfaction exactly, that is,
$|\widehat a_i\cap W_{\mathbf{o}}|=\mathrm{sat}_i(\mathbf{o})$. It
also handles repeated selections of the same original candidate, since
its occurrences in different rounds are distinct pairs.

The feasible choices, however, differ. A temporal outcome selects one
pair from each round, whereas an ordinary committee of size $\ell$ may
select several pairs from one round and none from another. This
difference can prevent every temporal outcome from satisfying
FJR$^\textit{\,mw}$+ in the constructed election. For example,
consider two voters and two rounds. In round~1, the voters approve
different candidates; in round~2, both approval sets are empty. Under
EJR$^\textit{\,mw}$+, each voter is entitled to one approved
candidate, but a temporal outcome can satisfy only one of them, while
a multiwinner committee can select both approved pairs from round~1.
Hence a multiwinner existence or computation result does not by itself
yield the corresponding temporal result.

The correspondence nevertheless transfers the earlier multiwinner
axioms to their temporal counterparts. First, if $W_{\mathbf{o}}$
satisfies FJR$^\textit{\,mw}$, then $\mathbf{o}$ satisfies temporal
FJR. To see this, fix a group $S$ and a set of rounds $T$, and choose
any $R\subseteq T$ with $|R|=\lfloor |T||S|/n\rfloor$. By the
definition of $\mu_S(T)$, some choice of candidates in $R$ gives every
member of $S$ satisfaction at least $\mu_S(T)$. The corresponding
candidate-round pairs form a comparison set of size
$|R|\leq\ell|S|/n$, so FJR$^\textit{\,mw}$ gives some member of $S$
satisfaction at least $\mu_S(T)$ in $W_{\mathbf{o}}$, and hence in
$\mathbf{o}$.

Second, if $W_{\mathbf{o}}$ satisfies EJR$^\textit{\,mw}$+, then
$\mathbf{o}$ satisfies temporal EJR+. Suppose that $\mathbf{o}$
violates temporal EJR+ for a $(\sigma,\tau)$-cohesive group $S$ and a
round $r$ in which all members of $S$ approve a common candidate $c$.
Since every voter meets a target of zero, the level
$\beta=\lfloor\sigma\tau/n\rfloor$ is positive, and all members of
$S$ have satisfaction below $\beta$. Moreover, $o_r$ is not approved
by every member of $S$, so the pair $(c,r)$ is unelected. Since
$n\beta\leq\sigma\tau\leq|S|\ell$, the pair $(c,r)$ has at least
$n\beta/\ell$ approvers below $\beta$, violating
\eqref{eq:mw-ejrplus}.

Consequently, FJR$^\textit{\,mw}$+ for $W_{\mathbf{o}}$ implies both
temporal FJR and temporal EJR+. As we show next, its relationship with
temporal FJR+ is different.

\subsection{Incomparability}

Throughout this subsection, we compare the two FJR+ axioms using the
construction \eqref{eq:mw-conversion}.

\begin{theorem}
\label{thm:mw-incomparable}
Temporal FJR+ does not imply FJR$^\textit{\,mw}$+ for
$W_{\mathbf{o}}$, and FJR$^\textit{\,mw}$+ for
$W_{\mathbf{o}}$ does not imply temporal FJR+.
Both statements hold even when every voter approves exactly one
candidate per round and the temporal outcome is Pareto efficient.
\end{theorem}

\begin{proof}
We give an example for each direction. Each table entry is the voter's
unique approved candidate in that round, and differently named
candidates are distinct within each example.

\emph{Temporal FJR+ does not imply FJR$^\textit{\,mw}$+.}
Consider five voters and three rounds:
\begin{center}
\begin{tabular}{cccc}
\toprule
Voter & Round 1 & Round 2 & Round 3\\
\midrule
1 & $a$ & $c_1$ & $c_1$\\
2 & $a$ & $c_2$ & $c_2$\\
3 & $a$ & $c_3$ & $c_3$\\
4 & $b$ & $c_4$ & $c_4$\\
5 & $b$ & $c_5$ & $c_5$\\
\bottomrule
\end{tabular}
\end{center}
The outcome $\mathbf{o}=(a,c_1,c_1)$ gives satisfactions
$(3,1,1,0,0)$. At level one, $U_1=\{4,5\}$; at levels two and three,
$U_2=U_3=\{2,3,4,5\}$. In each case,
\[
(h_1(U_\beta),h_2(U_\beta),h_3(U_\beta))=(2,1,1),
\]
whose sum is $4<5\beta$, so \cref{thm:verification} shows that
$\mathbf{o}$ satisfies temporal FJR+. In the multiwinner election,
however, voters~4 and~5 both approve the unelected candidate $(b,1)$
and have satisfaction zero. Their number exceeds $n/k=5/3$, which
violates EJR$^\textit{\,mw}$+ and hence FJR$^\textit{\,mw}$+. In
fact, the group $S=\{4,5\}$ and the set $T=\{(b,1)\}$ also violate
FJR$^\textit{\,mw}$. Finally, preserving voter~1's satisfaction of
three fixes the choice in every round, so $\mathbf{o}$ is Pareto
efficient.

Intuitively, voters~4 and~5 agree only in the first round, and
temporal FJR+ takes into account their lack of agreement in the later
rounds. In the multiwinner election, by contrast, their share of the
three committee seats entitles them to an approved candidate, even
though another candidate from the same round has already been
selected.

\emph{FJR$^\textit{\,mw}$+ does not imply temporal FJR+.}
Now consider eleven voters and three rounds:
\begin{center}
\begin{tabular}{cccc}
\toprule
Voters & Round 1 & Round 2 & Round 3\\
\midrule
$1,\ldots,7$ & $a$ & $b$ & $b$\\
8 & $a$ & $c_8$ & $c_8$\\
9 & $a$ & $c_9$ & $c_9$\\
10 & $c_{10}$ & $c_{10}$ & $c_{10}$\\
11 & $c_{11}$ & $c_{11}$ & $c_{11}$\\
\bottomrule
\end{tabular}
\end{center}
The outcome $\mathbf{o}=(a,c_{10},c_{11})$ gives every voter
satisfaction one. However,
\[
\sum_{r=1}^3 h_r(N)=9+7+7=23>22=2n,
\]
so condition~(i) of temporal FJR+ requires some voter to have
satisfaction at least two, and $\mathbf{o}$ violates the axiom.

We show that $W_{\mathbf{o}}$ nevertheless satisfies
FJR$^\textit{\,mw}$+ by ruling out the assignments in
\cref{def:mw-fjrplus}. At level one, every voter already has
satisfaction one, so all voter weights must be zero, contrary to
\eqref{eq:mw-proportionality}. Now fix $\beta\in\{2,3\}$. We bound the
total candidate weight from below by considering three disjoint sets of
candidate-round pairs. Voters $1,\ldots,7$ approve only the pairs
$(a,1)$, $(b,2)$, and $(b,3)$. Since each candidate weight bounds the
assignment to each of these seven voters, summing gives
\[
\beta\sum_{i=1}^7z_i
\leq7\bigl(y_{(a,1)}+y_{(b,2)}+y_{(b,3)}\bigr).
\]
For each $i\in\{8,9\}$, the elected pair $(a,1)$ contributes at most
$z_i$, so the weights of voter $i$'s other two approved pairs,
$(c_i,2)$ and $(c_i,3)$, total at least $(\beta-1)z_i$; these pairs
are approved only by voter $i$. Likewise, every pair approved by
voter~10 or~11 is approved by that voter alone, so these pairs require
total candidate weight at least $\beta(z_{10}+z_{11})$. As the three
bounds concern disjoint sets of pairs (the weight of $(a,1)$ is counted
only in the first), we may add them to obtain
\begin{align*}
\sum_{c\in\widehat C}y_c \geq \frac{\beta}{7}\sum_{i=1}^7z_i
 +(\beta-1)(z_8+z_9)+\beta(z_{10}+z_{11}) \geq \frac{2}{7}\sum_{i=1}^{11}z_i,
\end{align*}
where the second inequality uses $\beta\geq2$ and the nonnegativity of
the voter weights. But then \eqref{eq:mw-proportionality} would give
\[
\sum_{i=1}^{11}z_i
\geq\frac{11}{3}\sum_{c\in\widehat C}y_c
\geq\frac{22}{21}\sum_{i=1}^{11}z_i,
\]
which is impossible when the total voter weight is positive.

To see that $\mathbf{o}$ is Pareto efficient, first consider an
outcome that does not select $a$ in round~1. Each of its three choices
satisfies at most one of voters $8,9,10,11$, so one of these voters
has satisfaction zero. If $a$ is selected in round~1, then keeping
voters~10 and~11 at positive satisfaction requires the other two
choices to be $c_{10}$ and $c_{11}$, one each. Every voter then again
has satisfaction one, so no Pareto improvement exists.

In this example, temporal FJR+ requires some voter to reach the integer
levels justified by the electorate's best attainable average
satisfaction. The multiwinner axiom instead asks whether a weighted
group can give each of its members the target level within its
proportional share, and the calculation above shows that no such
weighted group exists, even though the temporal total exceeds $2n$.
\end{proof}

The two examples separate the definitions without relying on empty
ballots, Pareto improvements, or the boundary conditions of temporal
FJR+. In particular, FJR$^\textit{\,mw}$+ can hold even when
condition~(i) of temporal FJR+ is violated.

\subsection{Static Preferences}

With static preferences, temporal outcomes correspond to outcomes of
approval-based apportionment, in which a candidate may be selected
several times \citep{brill2024approval}. The candidate copies in
\eqref{eq:mw-conversion} then give the standard reduction from
apportionment to multiwinner voting. \citet{golz2026approval} show
that FJR, EJR, and EJR+ coincide in this setting. We extend this
picture to temporal FJR+ and multiwinner FJR+.

\begin{proposition}
\label{prop:mw-static}
For static preferences, an outcome $\mathbf{o}$ satisfies
temporal FJR+ if and only if $W_{\mathbf{o}}$ satisfies
FJR$^\textit{\,mw}$+ in the election of
\eqref{eq:mw-conversion}. This holds even if $\mathbf{o}$
selects the same original candidate more than once.
\end{proposition}

\begin{proof}
Fix a level $\beta\in[\ell]$. Since preferences are static,
$h_r(U_\beta)=h_1(U_\beta)$ in every round. As all positive round
values are then equal, equality with a positive total is excluded, so
by \cref{thm:verification}, temporal FJR+ holds exactly when
\begin{equation}
\ell h_1(U_\beta)<n\beta
\label{eq:mw-static}
\end{equation}
at every level.

Suppose first that temporal FJR+ holds, and consider assignments that
would violate FJR$^\textit{\,mw}$+ at level $\beta$. Voters outside
$U_\beta$ have $z_i=0$; since their nonnegative assignments sum to
$\beta z_i=0$, all of these assignments are zero. Every
candidate-round pair is approved by at most $h_1(U_\beta)$ voters in
$U_\beta$, so
\[
\beta\sum_{i\in N}z_i
=\sum_{c\in\widehat C}\sum_{i\in U_\beta}x_{ic}
\leq h_1(U_\beta)\sum_{c\in\widehat C}y_c.
\]
By \eqref{eq:mw-proportionality},
$\sum_{c\in\widehat C}y_c\leq(\ell/n)\sum_{i\in N}z_i$. Since the
total voter weight is positive, the displayed inequality implies
$n\beta\leq\ell h_1(U_\beta)$, contradicting \eqref{eq:mw-static}.

Conversely, suppose that temporal FJR+ is violated at level $\beta$,
so that $h_1(U_\beta)\geq n\beta/\ell$. Choose an original candidate
$c$ approved by $h_1(U_\beta)$ voters in $U_\beta$. At least one copy
$(c,r)$ of $c$ is unelected, since selecting all $\ell$ copies would
give these voters satisfaction $\ell\geq\beta$. This unelected copy
has at least $n\beta/\ell$ approvers below $\beta$, which violates
EJR$^\textit{\,mw}$+ and therefore FJR$^\textit{\,mw}$+.
\end{proof}

The equivalence relies on using a separate candidate for each round.
An alternative comparison keeps the original candidate set $C$: with
static preferences, an outcome selecting $\ell$ distinct candidates
gives an ordinary committee of size $\ell$ with the same
satisfactions. Temporal FJR+ still implies FJR$^\textit{\,mw}$+ under
this comparison, since the first half of the proof applies with $C$
in place of $\widehat C$. The converse, however, can fail. Consider
one voter, two rounds, and $C=\{a,b\}$, and let the voter's static
approval set be $\{a\}$. The committee $\{a,b\}$ satisfies
FJR$^\textit{\,mw}$+: at level one, no voter is below the target, and
at level two, the only approved candidate is elected, so
$2z_1=x_{1a}\leq z_1$ forces $z_1=0$. Yet the temporal outcome $(a,b)$
violates FJR+, since selecting $a$ twice would give satisfaction two.
Keeping the original candidates thus compares a model that permits
repeated selections with one that allows each candidate only once.
Even for static preferences, one must therefore specify which
comparison is intended.

\end{document}